\documentclass[11pt]{article}

\usepackage[
  letterpaper,
  margin=1in
]{geometry}

\usepackage[utf8]{inputenc}
\usepackage[T1]{fontenc}

\usepackage{lmodern}
\usepackage{microtype}

\usepackage{amsmath}
\usepackage{amsthm}
\usepackage{amssymb}
\usepackage{mathtools}
\usepackage{thmtools}
\usepackage{nicefrac}

\usepackage{graphicx}
\usepackage{booktabs}
\usepackage{tabularx}
\usepackage[table]{xcolor}
\usepackage{enumitem}
\usepackage{xifthen}

\usepackage[
  ruled,
  vlined,
  linesnumbered
]{algorithm2e}

\usepackage{algpseudocode}

\usepackage[numbers,sort&compress]{natbib}

\usepackage{url}
\usepackage{hyperref}

\hypersetup{
  colorlinks=true,
  linkcolor=blue!45!black,
  citecolor=green!35!black,
  urlcolor=blue!55!black,
  pdfauthor={},
  pdftitle={Attribute-based Undetectable Watermarking for Generative AI Models},
  pdfsubject={Cryptographic watermarking for generative AI},
  pdfkeywords={
    generative AI,
    watermarking,
    attribute-based watermarking,
    constrained pseudorandom functions
  }
}

\usepackage[capitalise]{cleveref}

\usepackage{multicol}
\newtheorem{theorem}{Theorem}

\newtheorem{lemma}[theorem]{Lemma}

\newtheorem{protoproposition}[theorem]{Proposition}

\newtheorem{protoquestion}{Question}

\theoremstyle{definition}

\newtheorem{definition}[theorem]{Definition}

\newtheorem{construction}{Construction}

\newcommand{\abs}[1]{\left\vert#1\right\vert}

\newcommand{\set}[1]{\left\{#1\right\}}
 \newcommand{\eps}{\varepsilon}

\newcommand{\wt}{\widetilde}

\newcommand{\zo}{\set{0,1}}

\newcommand{\bbN}{\mathbb{N}}

\newcommand{\cE}{\mathcal{E}}
\newcommand{\cC}{\mathcal{C}}
\newcommand{\cY}{\mathcal{Y}}

\newcommand{\cD}{\mathcal{D}}

\newcommand{\cF}{\mathcal{F}}

\newcommand{\cT}{\mathcal{T}}
\newcommand{\cB}{\mathcal{B}}

\newcommand{\cX}{\mathcal{X}}

\newcommand{\cW}{\mathcal{W}}
\newcommand{\cS}{\mathcal{S}}
\newcommand{\cP}{\mathcal{P}}

\newcommand{\adv}{\mathcal{A}}

\newcommand{\secpara}{\kappa}

\renewcommand{\Pr}{\operatorname*{\mathbf{Pr}}}
\DeclareMathOperator*{\E}{\mathbf{E}}

\newcommand{\pr}[2][]{ \ifthenelse{\isempty{#1}}
  {\Pr\left[#2\right]} {\Pr_{#1}\left[#2\right]} }
\newcommand{\ex}[2][]{ \ifthenelse{\isempty{#1}}
  {\E\left[#2\right]}
  {\E_{#1}\left[#2\right]} }

\newcommand{\poly}{\mathsf{poly}}
\newcommand{\negl}{\mathsf{negl}}

\newcommand{\ie} {i.e.,\ }
\newcommand{\eg} {e.g.,\ }

\newcommand{\sk}{\mathsf{sk}}

\newcommand{\msk}{\mathsf{msk}}
\newcommand{\csk}{\mathsf{csk}}
\newcommand{\dk}{\mathsf{dk}}
\newcommand{\pp}{\mathsf{pp}}

\newcommand{\prompt}{\pi}
\newcommand{\Setup}{\mathbf{Setup}}

\newcommand{\Constrain}{\mathbf{Constrain}}
\newcommand{\Generate}{\mathbf{Generate}}
\newcommand{\Issue}{\mathbf{Issue}}
\newcommand{\Detect}{\mathbf{Detect}}
\newcommand{\MasterDetect}{\mathbf{MasterDetect}}
\newcommand{\Gen}{\mathbf{Gen}}
\newcommand{\Eval}{\mathbf{Eval}}
\newcommand{\constrain}{\mathbf{Constrain}}
\newcommand{\CEval}{\mathbf{CEval}}
\newcommand{\Encode}{\mathbf{Encode}}
\newcommand{\Decode}{\mathbf{Decode}}

\newcommand{\PRC}{\mathbf{PRC}}
\newcommand{\CPRF}{\mathbf{CPRF}}
\newcommand{\PRG}{\mathbf{PRG}}

\newcommand{\Enc}{\mathbf{Enc}}
\newcommand{\Dec}{\mathbf{Dec}}

\newcommand{\dist}{\mathsf{dist}}

\newcommand{\out}{out}
\newcommand{\Model}{\mathbf{Model}}
\newcommand{\RandRecover}{\mathbf{RandRecover}}
\newcommand{\Classifier}{\mathbf{Classifier}}

\newcommand{\frakG}{\mathfrak{G}}
\renewcommand{\Game}{\frakG}

\title{\textbf{Attribute-based Undetectable Watermarking for Generative AI Models}}

\author{
  Miryam Mi-Ying Huang \thanks{Part of the work was done when Miryam is at University of Southern California}\\
  Carnegie Mellon University\\
  Pittsburgh, PA, USA\\
  \texttt{miyingh@andrew.cmu.edu}
  \and
  Chung-Wei Lee\\
  Independent Researcher\\
  \texttt{lee@chungwei.net}
  \and
  Max Raffel\\
  University of Southern California\\
  Los Angeles, CA, USA\\
  \texttt{mraffel@usc.edu}
  \and
  Er-Cheng Tang\\
  University of Washington\\
  Seattle, WA, USA\\
  \texttt{erchtang@uw.edu}
}

\date{}

\begin{document}

\maketitle

\begin{abstract}

Generative AI systems increasingly produce content whose provenance is difficult to verify, motivating watermarking techniques for identifying model-generated outputs. Existing cryptographic watermarking methods provide strong undetectability guarantees: without a detection key, watermarked outputs are computationally indistinguishable from unwatermarked ones. However, these approaches do not address the crucial deployment challenge of how to safely delegate detection capabilities. With an unrestricted detection key, a malicious detector may use the detection key beyond its intended scope, enabling watermark sanitization, scope abuse, and user profiling.

To mitigate this safety concern, we introduce, to the best of our knowledge, the first \emph{attribute-based watermarking} for generative AI models, providing fine-grained, policy-controlled watermark detection. In our approach, each generated output is associated with attributes, and each detection key is \emph{constrained by a policy} on potential attributes. A detection key can only be used to detect watermarked outputs whose attributes satisfy the corresponding policy, while watermarked outputs that fall outside the policy remain computationally indistinguishable from unwatermarked ones. We construct such an attribute-based watermarking scheme and formalize its security properties, including consistency, adaptive robustness to bounded corruptions, undetectability, and soundness, along with a security proof under standard cryptographic assumptions. Our construction integrates constrained pseudorandom functions, pseudorandom error-correcting codes, and randomness recovery procedures with generative AI models. Finally, we implement a prototype and an empirical evaluation, demonstrating that attribute-based watermarking is both effective and practical.

\end{abstract}

\section{Introduction}
\label{sec:introduction}

Generative artificial intelligence systems now produce text, images, audio, and video that are increasingly difficult to distinguish from human-created content. Large language models (LLMs) achieve high levels of fluency and instruction following, while generative vision models produce realistic synthetic media at scale. These capabilities create substantial societal and economic value, but they also make content provenance a pressing technical and policy challenge. Governments, platforms, academic institutions, and model providers increasingly need mechanisms to identify synthetic content, authenticate digital media, and enforce AI-use policies.

Recent policy efforts have explicitly identified watermarking, content labeling, and provenance tracking as tools for reducing the risks of synthetic content~\cite{executiveorder2023ai,nistAISIC2024,nistSyntheticContent2024}. In particular, the U.S. AI Safety Institute Consortium lists watermarking synthetic content among its priority activities, while NIST's synthetic-content research agenda highlights unresolved challenges concerning watermark robustness and security gaps in content-authentication ecosystems~\cite{nistAISIC2024,nistSyntheticCall2024}. These efforts underscore both the practical importance of watermarking and the need for stronger, deployable security guarantees.

At the same time, real-world institutions have begun deploying ad hoc mechanisms to detect improper AI use. For example, ICML 2026 used hidden instructions embedded in submitted PDFs to detect violations of its LLM-use policy for peer review~\cite{icml2026llmreview}. This approach identified hundreds of violations, but the organizers also noted that it was easy to circumvent once known: a reviewer could remove the hidden instruction by sanitizing the PDF, extracting the visible text, or otherwise changing the document-processing pipeline before using an LLM. This example highlights the fragility of non-cryptographic provenance tracking. Prompt-injection-based ``honeypots'' can detect careless misuse, but they do not provide a durable provenance mechanism. More generally, ad hoc approaches do not provide formal guarantees regarding robustness, authorization, or privacy.

Watermarking has emerged as a leading technical approach for provenance tracking in generative AI. A watermarking scheme embeds a statistically verifiable signal into generated outputs so that an authorized party can later determine whether an output was produced by the model. Early watermarking schemes for language models modify the sampling procedure by intentionally perturbing the model's output distribution~\cite{kirchenbauer2023watermark}, but this approach creates a tension between detectability, robustness, and preservation of generation quality~\cite{christ2024undetectable}.

Recent cryptographic work addresses this tension through undetectable watermarking~\cite{christ2024undetectable}. In an undetectable watermarking scheme, watermarked and unwatermarked outputs are computationally indistinguishable to any efficient party that does not hold the relevant detection key, even under adaptive querying~\cite{christ2024undetectable}. This property is stronger than preserving any particular quality metric: it implies that no efficient test can distinguish the two distributions without the key. Undetectable watermarking thus provides a principled way to support detection while avoiding observable degradation of generated content. Subsequent works improve undetectable watermarks in terms of the tolerable error rate \cite{christ2024pseudorandom} and the attack model in which the watermarks may remain effective to attackers who have access to either a detection oracle~\cite{alrabiah2025ideal} or a public detection key~\cite{christ2025improved}.

However, undetectability alone does not solve the authorization problem that arises when detection capability must be delegated. In many realistic deployments, a model provider should not hand out an all-purpose detector. Different detectors may be authorized to inspect different classes of outputs. Existing watermarking schemes largely treat detection as coarse-grained: once a detector receives detection capability, the scheme itself does not enforce fine-grained restrictions on which outputs the detector may test.

Consider a conference that wishes to delegate watermark detection to its area chairs. If detection is publicly available, an author could repeatedly edit an AI-generated submission and query the detector until detection fails, thereby using the detector as a sanitization oracle. A conventional secret-key watermark prevents such public access, but does not natively enforce fine-grained delegation: an area chair who receives a conference-wide detection key may also be able to test submissions outside
the chair's assigned subject area.

This lack of policy-scoped detection creates several concrete vulnerabilities.

\begin{itemize}[leftmargin=*]
    \item \textbf{Watermark sanitization.}
    If the detection key is broadly accessible, an adversary may use the detection key to sanitize the watermark. For example, a malicious influencer generating AI-written medical, financial, or political content could repeatedly edit an output and apply detection until the watermark is no longer detected.

    \item \textbf{Scope abuse.}
    A detector authorized in one setting may apply detection outside that setting. For example, an enterprise compliance team may be authorized to verify whether externally released reports were generated by an AI system, but an unrestricted detector could also be applied to internal employee documents, legal drafts, or confidential strategy memos. The detector would then obtain provenance information outside the scope of its delegated authority.

    \item \textbf{User profiling.}
    Broad detection capability can leak behavioral information. If arbitrary parties can scan a user's public posts for watermarks, they may infer which topics the user relies on AI assistance for, enabling targeted phishing, manipulation, or reputational attacks.
\end{itemize}

These examples suggest a common failure mode: existing detection mechanisms only authenticate \emph{whether} content is watermarked, but do not cryptographically enforce \emph{who} may test it, \emph{which} content they may test, or \emph{under what policy}. Thus, watermarking should provide not only undetectability, but also cryptographic access control for detection. We introduce \emph{attribute-based watermarking} for generative AI. In our framework, each generated output is associated with an attribute, such as a topic, domain, course identifier, policy category, or other application-defined label. A model provider can issue a detection key tied to a policy over the attributes. The key enables detection only for outputs whose attributes satisfy the policy, while preserving undetectability outside the authorized scope. Thus, verification becomes fine-grained and policy-controlled rather than all-or-nothing.

In the conference example, each generated output is associated with a subject-area attribute, and each area chair receives a detection key constrained to the policy for the chair's assigned area. Authors and ordinary reviewers receive no detection capability, while an area chair's key reveals no useful information about watermarks on out-of-area submissions. This confines the consequences of key leakage or misuse to the delegated scope. It
does not prevent an authorized area chair from misusing detection on an in-scope submission; as in conventional access-control systems, such trusted-role misuse lies outside the goal of policy-scoped delegation.

At a technical level, our construction binds the randomness used for watermark generation to the output attribute through a constrained pseudorandom function. The master secret key allows the model provider to generate watermarked outputs for arbitrary attributes. A delegated detection key corresponds to a policy and allows evaluation of the relevant pseudorandom value only on authorized attributes. We combine constrained pseudorandom functions with pseudorandom error-correcting codes and randomness recovery procedures, following the cryptographic watermarking paradigm for randomized generative models. This yields a scheme in which authorized detection remains robust to small perturbations, whereas unauthorized detection reveals no useful information about the watermark.
Our contributions are as follows.

\begin{itemize}[leftmargin=*]
    \item \textbf{The first attribute-based verification framework for generative AI.}
    We introduce a modular framework for watermarking schemes with policy-scoped detection. The framework separates master generation authority from delegated detection authority while modeling detection keys as constrained by policies over output attributes. Since the framework is modular rather than tied to a fixed instantiation, advances in any underlying primitive immediately translate to improvements in the resulting watermarking scheme.

    \item \textbf{Security notions for delegated watermark detection.}
    We define consistency for authorized detection, adaptive robustness against bounded modifications of watermarked outputs, undetectability outside the authorized policy scope, and soundness against adversaries that do not hold the master secret key. These notions capture the distinction between detecting watermarks within an authorized domain and learning useful information about them outside that domain.

    \item \textbf{A cryptographic construction from constrained pseudorandom functions and pseudorandom codes.}
    We give a construction that combines constrained pseudorandom functions, pseudorandom error-correcting codes, pseudorandom generators, and randomness recovery for generative models. The constrained pseudorandom function enforces policy-scoped access to the pseudorandom code, which provides a robust and undetectable watermark. 

    \item \textbf{Security proofs of our watermarking scheme.}
    We prove that the construction satisfies the proposed security properties under standard assumptions on the underlying primitives, including correctness and security of the constrained pseudorandom function, pseudorandomness and robustness of the pseudorandom code, and compatibility of the model's randomness recovery procedure with the tolerated output perturbations.
    \item \textbf{Prototype implementation and empirical evaluation.}
    We implement a prototype of attribute-based watermarking for generative AI models and evaluate its behavior. Our implementation includes master-key generation, constrained detection-key issuance, watermark generation, and authorized detection. Empirically, we validate selective detectability under authorized policies and measure the overhead introduced by the cryptographic components.
\end{itemize}

Overall, attribute-based watermarking provides a cryptographic foundation for controlled provenance tracking in generative AI. It preserves the central benefit of undetectable watermarking, namely that watermarked outputs remain indistinguishable from ordinary model outputs to unauthorized parties, while adding a missing access-control layer: detection can be delegated according to explicit policies over the attributes of generated content.

\section{Preliminary}
\label{sec:prelim}

This section provides abstractions for generative AI models and classifiers, and reviews relevant cryptographic primitives that will be used in our watermarking construction.

The following notions will be used. We denote by $\negl(n)$ an unspecified function $f: \bbN \to \mathbb{R}_{\ge 0}$ that decays asymptotically faster than any inverse polynomial, or equivalently, a function with $f(n) = o(\frac{1}{n^c})$ for every $c>0$. For two strings $u, v \in \Sigma^n$ over alphabet $\Sigma$, we define their relative Hamming distance \(\dist(u,v) := \tfrac{|\set{i : u_i \neq v_i}|}{n}\) as the fraction of coordinates on which the two strings differ. We also define $\dist(u,v)=\infty$ when $u,v$ have different length. For a set $S$, we define $\dist(u,S) := \min_{v \in S} \dist(u,v)$. Given a randomized algorithm $R$, we denote its execution on an input $x$ either by $R(x)$, which uses an internal randomness, or $R(x; r)$, where $r$ is the randomness written out explicitly.

\subsection{Generative AI Models}
\label{sec:model}
Generative AI systems, such as large language models, are inherently randomized: when run on the same prompt multiple times, the model may produce different outputs based on internally sampled randomness. 
We abstract such systems as randomized algorithms that map prompts to outputs.
\begin{definition}
\label{def:AI-model}
    A generative AI model is a randomized algorithm $\Model$ that takes as input a prompt $\prompt$ and generates $\Model(\prompt; r) \to \out \in \Sigma^n$ an output of length $n$, where $r$ denotes randomness used by the algorithm.
\end{definition}

Recent works \cite{christ2024pseudorandom,gunn2025an} observed that the randomness used by generative AI models can often be approximately recovered from the generated output itself.
Intuitively, the generation process preserves substantial information about the underlying latent randomness, and the generated output exhibits a high correlation with the randomness used, especially when the output has sufficiently high entropy. Consequently, one can approximately recover a noisy version of the original randomness. Moreover, this approximate recovery can remain possible even if the output is slightly modified or corrupted. We refer the readers to \cite{christ2024pseudorandom,gunn2025an} for more information on how randomness recovery is achieved.

\begin{definition}
\label{def:randomness-recovery}
    A generative AI model is said to admit robust randomness recovery with constant error rate $\delta$ if there is an efficient deterministic algorithm $\RandRecover$ that takes as input a (potentially corrupted) generated output $\out' \in \Sigma^n$ and produces a recovered randomness $r' \in \zo^n$ that is close to the truncated original randomness $r_{[1:n]} \in \zo^n$ in relative Hamming distance. That is, for any error channel $\cE$,
    \[\Pr_{r}\left[ \dist(\out',\out) < \delta \implies \dist(r', r_{[1:n]}) < 2\delta \;\middle|\; \begin{matrix}
        \out\gets \Model(\prompt;r)\\
        \out'\gets \cE(\out)\\
        r' \gets \RandRecover(\out')
    \end{matrix}\right] = 1 - \negl(n)\]
    as long as the output distribution $\Model(\prompt)$ has sufficiently high entropy on prompt $\prompt$. For convenience, throughout this work, we will only consider prompts whose output distribution has sufficiently high entropy.
\end{definition}

Next, we introduce the role of classifiers in our framework. Given a generative AI model, we use a classifier to assign attributes to generated outputs. A key requirement is that these attribute assignments remain stable across multiple samples generated from the same prompt. In particular, we expect that two independently generated outputs under the same prompt are likely to be assigned the same attribute. For generative AI models, such consistency may not hold automatically, but can often be achieved empirically, for example by applying strategies such as chain-of-thought that encourage more stable and semantically consistent outputs.
Generated outputs are also expected to retain stable attributes even under small perturbations. This requirement is closely related to adversarial robustness in machine learning, where a classifier is expected to produce consistent outputs under small adversarial modifications. We formalize these requirements below via notions of adversarial robustness and consistency with respect to the generative model.

\begin{definition}
    A classifier is an algorithm $\Classifier(m) \to x$ that maps a message $m \in \Sigma^n$ to a label $x \in \zo^{n'}$. 
    A classifier is $\eta$-consistent on a generative AI model $\Model$ if for every prompt $\prompt$,
    \[\Pr_{m,m' \gets \Model(\prompt)}\left[ \Classifier(m') = \Classifier(m) \right] \ge \eta.\]
    A classifier is $(\delta,\eps)$-adversarially robust over a distribution $\cD$ if for every probabilistic polynomial-time adversary $\adv$,
    \[\Pr_{m \gets \cD,\; m' \gets \adv(m)} \left[ \Classifier(m') \neq \Classifier(m) \;\land\; \dist(m',m) < \delta \right] \le \eps.\]
\end{definition}
If the classifier is adversarially robust over the distribution $\Model(\prompt)$ for every $\prompt$, we simply say that $\Classifier$ is an adversarially robust output classifier for $\Model$.
In our analysis, we consider $\delta$ being a small constant, $\eps = \negl(n)$, and $\eta = 1 - \negl(n)$.

\subsection{Pseudorandom Error-Correcting Codes}
Traditional error-correcting codes are designed to recover messages from noisy transmissions. 
Pseudorandom codes (PRC) additionally require the encoded codewords to appear computationally indistinguishable from uniformly random strings while still supporting reliable decoding.

\begin{definition}[Pseudorandom Code \cite{christ2024pseudorandom,alrabiah2025ideal}]
    A secret-key PRC over alphabet $\Sigma$ with threshold $\delta$ is a tuple of probabilistic polynomial time algorithms $(\Gen, \Encode, \Decode)$ with the following syntax.
    \begin{itemize}
            \item $\Gen(1^\secpara; r) \to \sk$: On input security parameter $\secpara$ and randomness $r$, it outputs a secret key $\sk$.
            \item $\Encode(\sk) \to c \in \Sigma^{n}$: On input a secret key $\sk$, it outputs a codeword $c$.
            \item $\Decode(\sk, c \in \Sigma^n) \to 1 / 0$: 
            On input a secret key $\sk$ and a word $c$, it outputs $1$ if the decoding succeeds; otherwise, it outputs $0$.
    \end{itemize}
    We refer the readers to \cref{PRCdef} for detailed definitions including $\delta$-adaptive robustness, soundness, and pseudorandomness. Pseudorandom codes have been constructed assuming sub-exponential hardness of learning parity with noise (LPN) \cite{christ2024pseudorandom,alrabiah2025ideal}.
\end{definition}

\subsection{Pseudorandom Generators (PRGs)}
\begin{definition}[Pseudorandom Generator]
A \emph{pseudorandom generator} (PRG) is a deterministic polynomial-time algorithm
$\PRG : \zo^{\ell} \to \zo^{n}$, where $n > \ell$, such that for every probabilistic polynomial-time adversary $\adv$,
\[
\left| 
\Pr_{s \gets \zo^{\ell}} \left[ \adv(\PRG(s)) = 1 \right]
-
\Pr_{y \gets \zo^{n}} \left[ \adv(y) = 1 \right]
\right| = \negl(\secpara).
\]
\end{definition}
\subsection{Constrained Pseudorandom Functions (CPRFs)} 
Constrained pseudorandom functions (CPRFs) extend ordinary pseudorandom functions by allowing the holder of a constrained key to evaluate the function only on a restricted subset of inputs \cite{boneh2013constrained}. 
Intuitively, constrained keys reveal limited functionality of the underlying PRF while still hiding information about evaluations outside the authorized region.

\begin{definition}[Constrained Pseudorandom Functions] Let $\secpara$ be security parameter and $\pp$ be a public parameter as inputs implicit in all algorithms. A Constrained Pseudorandom Function (CPRF) with key space $\mathcal{K}$, domain $\mathcal{X}$, and range $\mathcal{Y}$ that supports constraints represented by the class of circuits $\mathcal{C}$, where each $C \in \cC$ maps $\mathcal{X}$ to $\{0,1\}$, is a tuple of polynomial time algorithms $(\Gen, \Eval, \constrain, \CEval)$ with the following syntax.

    \begin{itemize}
        \item $\Gen(1^\secpara) \rightarrow \msk$: The randomized key generation algorithm takes as input a security parameter $\secpara$ and outputs a master secret key $\msk \in \mathcal{K}$.
        \item $\Eval(\msk,x)\rightarrow y$: The deterministic evaluation algorithm takes as input the master secret key $\msk$ and input $x \in \mathcal{X}$, and outputs $y \in \mathcal{Y}$.
        \item $\constrain(\msk, C)\rightarrow \csk$: The randomized constrain algorithm takes as input the master secret key $\msk$ and a constraint circuit $C \in \mathcal{C}$, and outputs a constrained key $\csk$.
        \item $\CEval(\csk,x)\rightarrow y$: The deterministic constrained evaluation algorithm takes as input the constrained key $\csk$ and an input $x \in \mathcal{X}$ and outputs $y \in \mathcal{Y}$.
    \end{itemize}

It is required that  $\Eval(\msk,x)$ and $\CEval(\csk,x)$ produces the same value when $C(x) = 1$, and that the value of $\Eval(\msk,x)$ is indistinguishable from random to any efficient adversary with $\csk$ when $C(x) = 0$. We refer the readers to \cref{CPRFdef} for the full definition.

CPRFs have been constructed for several circuit classes, such as general circuits, inner-product circuits, and logarithmic depth circuits \cite{brakerski2015constrained,couteau2023constrained,servan2024constrained}. These results are based on standard cryptographic assumptions, such as the decisional Diffie-Hellman (DDH) or the learning with errors (LWE) assumptions, with some of the assumptions commonly believed to be quantum-secure (post-quantum).
\end{definition}

\section{Attribute-based Watermarking: Definition and Construction}
\label{sec:watermarking}
Attribute-based watermarking associates generated outputs with attributes while allowing watermark detection only under authorized policies. Such a framework is particularly relevant for generative AI systems where watermark verification, robustness, selective detectability, and output quality preservation may all be required simultaneously.

This section introduces the formal framework and security requirements for attribute-based watermarking schemes for generative AI models. The security notions characterize robustness against adversarial corruptions, indistinguishability from ordinary model outputs, and soundness under the presence of a detection key that is authorized for other attributes. A corresponding construction is presented in \cref{sec:construction} and its security analysis is given in Appendix \ref{sec:proofs}.

Let $\Model$ be a generative AI model and $\Classifier$ be an adversarially robust output classifier for the model as defined in \cref{sec:model}.
Given a prompt $\prompt$, the model generates an output $\out \gets \Model(\prompt)$. 
We will apply the classifier to obtain an attribute $x \gets \Classifier(\out)$ associated with the output. 
A policy $f \in \cF$ specifies a predicate over the attribute space $\cX$, and determines whether detection is authorized for outputs with attribute $x$ by evaluating $f(x) \in \zo$. In an attribute-based watermarking scheme, a detection key for $f$ only enables detection on outputs whose attribute satisfies $f(x)=1$.

The attribute serves as an intermediate description of the output, and policies operate on attributes rather than directly on raw outputs.
For example, an attribute $x \in \cX$ can be a set of keywords of the output, a policy $f$ can be a function that determines if an attribute $x \in \cX$ contains a keyword that belongs to a certain area, and a policy family $\cF$ can contain policies of the above type for several different areas (\eg medical, education, finance).

\begin{definition}
    An \emph{attribute-based watermarking scheme} $\cW$ with attribute space $\cX$ and policy family $\cF$ for a generative AI model $\Model$ consists of the following efficient algorithms:
\begin{itemize}
    \item $\Setup(1^\secpara) \to \msk$: On input a security parameter $\secpara$, it outputs a master secret key $\msk$.
    \item $\Issue(\msk, f) \to \dk_f$: On input a master secret key $\msk$ and a policy $f \in \cF$, it outputs a detection key $\dk_f$ tailored for the policy $f$.
    \item $\Generate(\msk, \prompt) \to \out$: On input a master secret key $\msk$ and a prompt $\prompt$, it generates an outcome $\out$.
    \item $\Detect(\dk_f, \out) \to 0/1$: On input a detection key $\dk_f$ and an outcome $\out$, it determines whether $\out$ is generated by the watermarking scheme.
    \item $\MasterDetect(\msk,\out) \to 0/1$: On input a master secret key $\msk$ and an outcome $\out$, it determines whether $\out$ is generated by the watermarking scheme.
\end{itemize}
\end{definition}

We require that an authorized detection behaves consistently with the master detection procedure. 
In particular, whenever a policy authorizes detection for an attribute, both procedures should produce the same detection outcome.

\begin{definition}[Consistency]
\label{def:consistency}
An attribute-based watermarking scheme $\cW$ is said to satisfy \emph{consistency}
if for every policy
$f \in \cF$ and every outcome $\out$ with $f(\Classifier(\out)) = 1$,
\[\Pr\left[
\Detect(\dk_f,\out)
=\MasterDetect(\msk,\out)\;\middle|\;\begin{matrix}
\msk \gets \Setup(1^\secpara)\\
\dk_f \gets \Issue(\msk,f)
\end{matrix}\right]\ge 1-\negl(\secpara)\]   
\end{definition}

We now characterize several security properties of our attribute-based watermarking scheme for generative AI models. These security properties hold even against adversaries who can get a detection key of a policy $f \in \cF$. 

The following notion captures the idea that a small fraction of modifications to a watermarked output $\out$ should not remove its watermark, even for modifications carried out by an adversary who observed many watermarked outputs and obtained a detection key $\dk_f$, as long as the detection key is unauthorized for detecting $\out$, \ie when $f$ satisfies $f(\Classifier(\out))= 0$.

\begin{definition}{(Adaptive Robustness with Detection)}
\label{def:adaptive robustness}
An attribute-based watermarking scheme $\cW$ is said to satisfy
\emph{adaptive robustness with detection against $\delta$-fraction of adversarial corruptions} if for every policy $f \in \cF$ and every probabilistic polynomial-time
adversary $\adv$,
\[
\Pr\left[
\frakG^{\mathsf{Robust}}_{\secpara,f,\cW}(\adv)=1
\right] = \negl(\secpara),
\]
where $\frakG^{\mathsf{Robust}}_{\secpara,f,\cW}(\adv)$ is defined as the following security game:
\begin{itemize}
    \item Generate $\msk \gets \Setup(1^\secpara)$ and send the detection key $\dk_{f} \gets \Issue(\msk,f)$ to $\adv$.
    \item $\adv$ can freely choose to query the generation algorithm $\out \gets \Generate(\msk,\prompt)$ on demand using any prompt $\prompt$ . The security game will record all the generated watermarked outputs as a list $\cS$.
    \item $\adv$ produces an output $\out^*$. The game outputs $1$ (\ie $\adv$ wins the security game) if
    \[f(\Classifier(\out^*)) = 0,\; \dist(\out^*, \cS)<\delta,\;\text{and }\MasterDetect(\msk,\out^*)= 0\]
\end{itemize}
The adversary wins the security game if it manages to produce an output $\out^*$ that is $\delta$-close to some watermarked sample with an attribute that does not satisfies the policy $f$, yet $\MasterDetect$ fails to recognize $\out^*$ as watermarked. Adaptive robustness with detection says that the winning probability of any efficient adversary is negligible. In other words, no efficient adversary can remove the watermark under a small perturbation using the information given by the detection key $\dk_f$ when the attribute of the watermarked output does not satisfy the policy $f$.
\end{definition}

The following notion captures that watermarking should not affect the quality of generated outputs. 
In particular, an adversary who holds an unauthorized detection key should not be able to distinguish watermarked outputs from ordinary model outputs.

\begin{definition}[Undetectability] 
\label{def:abundetectability}
    $\cW$ is said to satisfy \emph{attribute-based undetectability} if for every policy $f \in \cF$ and every probabilistic polynomial-time adversary $\adv$,
    \[\abs{\Pr\left[1 \gets \adv^{\mathsf{Env}(1^\secpara, f)}\right] - \Pr\left[1 \gets \adv^{\mathsf{Env}'(1^\secpara, f)}\right]} = \negl(\secpara),\]

where $\mathsf{Env}(1^{\secpara}, f)$ and $\mathsf{Env}'(1^{\secpara}, f)$ are defined as follows.

$\mathsf{Env}(1^{\secpara}, f)$:
\begin{itemize}
    \item Generate $\msk \gets \Setup(1^\secpara)$ and return a detection key $\dk_{f} \gets \Issue(\msk,f)$.
    \item Upon query a prompt $\prompt$, return a generated output $\out \gets \Generate(\msk, \prompt)$.
\end{itemize}
$\mathsf{Env'}(1^{\secpara})$:
\begin{itemize}
    \item Generate $\msk \gets \Setup(1^\secpara)$ and return a detection key $\dk_{f} \gets \Issue(\msk,f)$.
    \item Upon query a prompt $\prompt$, compute an ordinary model output $\out \gets \Model(\prompt)$.\\ If $f(\Classifier(\out)) = 0$, return $\out$. Otherwise return $\out' \gets \Generate(\msk, \pi)$.
\end{itemize}
\end{definition}

The following property says that there is no way for a person who holds a detection key $\dk_f$ to come up with an output with an attribute that does not satisfy the policy $f$, but the output gets identified as AI generated. 

\begin{definition}[Soundness]
\label{def:soundness}
$\cW$ is said to satisfy \emph{soundness} if for any probabilistic polynomial-time adversary $\adv$,
\[
\Pr\left[
\mathfrak{G}^{\mathsf{Sound}}_{\secpara, f, \cW}(\adv)=1
\right] = \negl(\secpara),
\]
where $\mathfrak{G}^{\mathsf{Sound}}_{\secpara, f, \cW}(\adv)$ is defined as the following security game:
\begin{itemize}
    \item Generate $\msk \gets \Setup(1^\secpara)$ and send the detection key $\dk_{f} \gets \Issue(\msk,f)$ to $\adv$.
    \item $\adv$ produces an output $\out^*$. The game outputs $1$ (\ie $\adv$ wins the security game) if
    \[f(\Classifier(\out^*)) = 0\; \text{and }\MasterDetect(\msk,\out^*)= 1\]
\end{itemize}
The adversary wins the security game if it manages to come up with an output $\out^*$ with an attribute that does not satisfies the policy $f$ without accessing the watermark generation algorithm, yet $\MasterDetect$ recognizes the output $\out^*$ as watermarked.
\end{definition}

\subsection{Construction}
\label{sec:construction}

We now present our construction for attribute-based watermarking. 
At a high level, the main challenge is to support fine-grained policy-scoped detection. A naive approach would be to generate a separate watermarking and detection key pair for every possible attribute, and distribute only the detection keys corresponding to authorized attributes. However, such an approach scales poorly with the dimension of the attribute space. For example, supporting a $50$-dimensional binary attribute space would require generating and storing $2^{50}$ distinct keys, resulting in prohibitive storage and key-management costs. Our construction avoids this exponential blowup by using constrained pseudorandom functions to derive attribute-dependent watermarking randomness from a single master secret key. Detection keys are issued with respect to policies over the attribute space, rather than individual attributes, enabling compact policy-scoped delegation. As a result, the total key size remains compact while incurring only a small computational overhead during detection.

Our construction  makes use of the following ingredients introduced in \cref{sec:prelim}, including a model $\Model$ with randomness recovery $\RandRecover$, an adversarially robust output classifier $\Classifier$ for the model, a $\PRG$, a $\CPRF = (\Gen, \Eval, \Constrain, \CEval)$, and a $\PRC = (\Gen, \Enc,\Dec)$. 

\begin{construction}\label{construction}
\noindent
\newline
\begin{minipage}{\linewidth}
\begin{multicols}{2}
\begin{itemize}
    
    \item[-] $\Setup(1^\secpara)$:  
    \begin{enumerate}
        \item $\msk \gets \CPRF.\Gen(1^\secpara)$.
        \item Output $\msk$.
    \end{enumerate}
    
    \item[-] $\Issue(\msk, f)$: 
    \begin{enumerate}
        \item $\dk \gets \CPRF.\Constrain(\msk, f)$.
        \item Output $\dk$.
    \end{enumerate}
    
    \item[-] 
    {$\Generate(\msk,\prompt)$:
    \begin{enumerate}
        \item $\out' \gets \Model(\prompt)$.
        \item $x \gets \Classifier(\out')$.
        
        \item $r\gets\CPRF.\Eval(\msk, x)$.
        \item $s \gets \PRC.\Gen(1^\secpara; \PRG(r))$.
        
        \item $c \gets \PRC.\Enc(s)$.
        \item $\out \gets \Model(\prompt; c)$.
        \item Output $\out$.
    \end{enumerate}
    }
    
    \columnbreak
    
    \item[-] 
    {$\Detect(\dk,\out)$:
    \begin{enumerate}
        \item $x \gets \Classifier(\out)$.
        
        \item $r \gets \CPRF.\CEval(\dk, x)$.
        \item $s \gets \PRC.\Gen(1^\secpara; \PRG(r))$.
        
        \item $c \gets \RandRecover(\out)$.
        \item $d \gets \PRC.\Dec(s, c)$.
        
        \item Output $1$ iff $d = 1$.
    \end{enumerate}
    }
    \item[-] 
    {$\MasterDetect(\dk,\out)$:
    \begin{enumerate}
        \item $x \gets \Classifier(\out)$.
        
        \item 
        $r \gets \CPRF.\Eval(\msk, x)$.
        \item $s \gets \PRC.\Gen(1^\secpara; \PRG(r))$.
        
        \item $c \gets \RandRecover(\out)$.
        \item $d \gets \PRC.\Dec(s, c)$.
        
        \item Output $1$ if and only if $d = 1$.
    \end{enumerate}
    }
\end{itemize}
\end{multicols}
\end{minipage}
\end{construction}

We prove the following theorem for our construction. The proof is given in Appendix \ref{sec:proofs}.
\begin{theorem}\label{thm:main-theorem}
    \cref{construction} is an attribute-based watermarking scheme that satisfies consistency (\cref{def:consistency}), adaptive robustness with detection against a constant $\delta$ fraction of adversarial corruptions (\cref{def:adaptive robustness}), undetectability (\cref{def:abundetectability}), and soundness (\cref{def:soundness}).
\end{theorem}

\section{Experiments}
\label{sec:experiments}

\paragraph{Implementation details.}

We implement our prototype in PyTorch~\cite{paszke2019pytorch}
using the Hugging Face Transformers
library~\cite{wolf-etal-2020-transformers}.
For text generation, we use
\texttt{meta-llama/Llama-3.2-1B-Instruct}%
~\cite{meta2024llama32,grattafiori2024llama3herdmodels}.
For all generation, we used \texttt{temperature} $= 1.0$, \texttt{top\_p} $= 0.95$, and \texttt{top\_k} $= 0.0$ (disabled) unless otherwise specified.
All experiments were conducted on NVIDIA L4 GPUs through Google
Colab and consumed approximately 70 compute credits in total.
Our implementation builds on the pseudorandom-code-based approach to
undetectable language-model watermarking introduced by Christ and
Gunn~\cite{christ2024pseudorandom}. The accompanying materials contain the complete implementation,
environment setup instructions, evaluation scripts, and scripts
for reproducing all figures and tables reported here%
\footnote{GitHub:
\url{https://github.com/maxraffel/attribute-based-watermarking}.}
We instantiate the attribute classifier using
\texttt{BAAI/bge-reranker-v2-m3}~\cite{chen-etal-2024-m3,baai2024bgererankerv2m3}. Given a generated output and a candidate attribute from a fixed vocabulary, the model assigns a relevance score to the corresponding output--attribute pair. The classifier is always used in a multi-label manner: we assign every attribute whose relevance score exceeds the fixed threshold $\tau=0.001$. We use the same classifier, fixed attribute vocabulary, and threshold throughout all experiments. The experimental conditions differ only in prompt construction: the single-attribute prompts are designed to elicit one target attribute, whereas the two-attribute prompts are designed to jointly elicit two target attributes.

Our implementation of \textsc{RandRecovery} uses probability-balanced token masks. At each watermarked generation step, the candidate tokens are partitioned into two sets corresponding to recovered bit values $0$ and $1$, with their aggregate model probabilities made as close as possible. These two values need not be equally likely under the original next-token distribution. Rather than conditioning on the selected mask at every step, the scheme uses the mask to truncate the token space only with a distribution-dependent probability, which is typically larger when the partition is more balanced and the next-token distribution has higher entropy. The complementary sampling rule is calibrated so that, after marginalizing over the watermarking randomness, the original model probability is unaffected. Thus, the next-token distribution is identical with and without watermarking; balancing the two probability masses instead maximizes how often the mask can be applied and thereby improves bit-embedding and recovery reliability.

During recovery, the token partitions are reconstructed deterministically by replaying the observed token sequence through the language model and recomputing the next-token probabilities at every position. Because the recovery algorithm does not receive the original prompt, generation maintains an auxiliary prompt-free model state and constructs the partitions from this state. The same state can subsequently be reconstructed using only the generated continuation.

We generate the first 100 tokens without embedding watermark bits. These tokens provide a warm-up prefix from which the prompt-free state can be initialized before watermark embedding begins. Consequently, watermarked generation requires an additional model forward pass at each decoding step to maintain the auxiliary state, while recovery requires model evaluations to reconstruct the masks associated with the observed tokens.

For computational efficiency, unwatermarked generation is performed in batches, and the model evaluations required during recovery are also batched. Unless stated otherwise, error bars show the empirical mean plus or minus two standard deviations, i.e., $\mathrm{mean}\pm 2\sigma$.

\begin{figure}[t]
    \centering
    \begin{minipage}{0.47\textwidth}
        \centering
        \includegraphics[width=\linewidth]{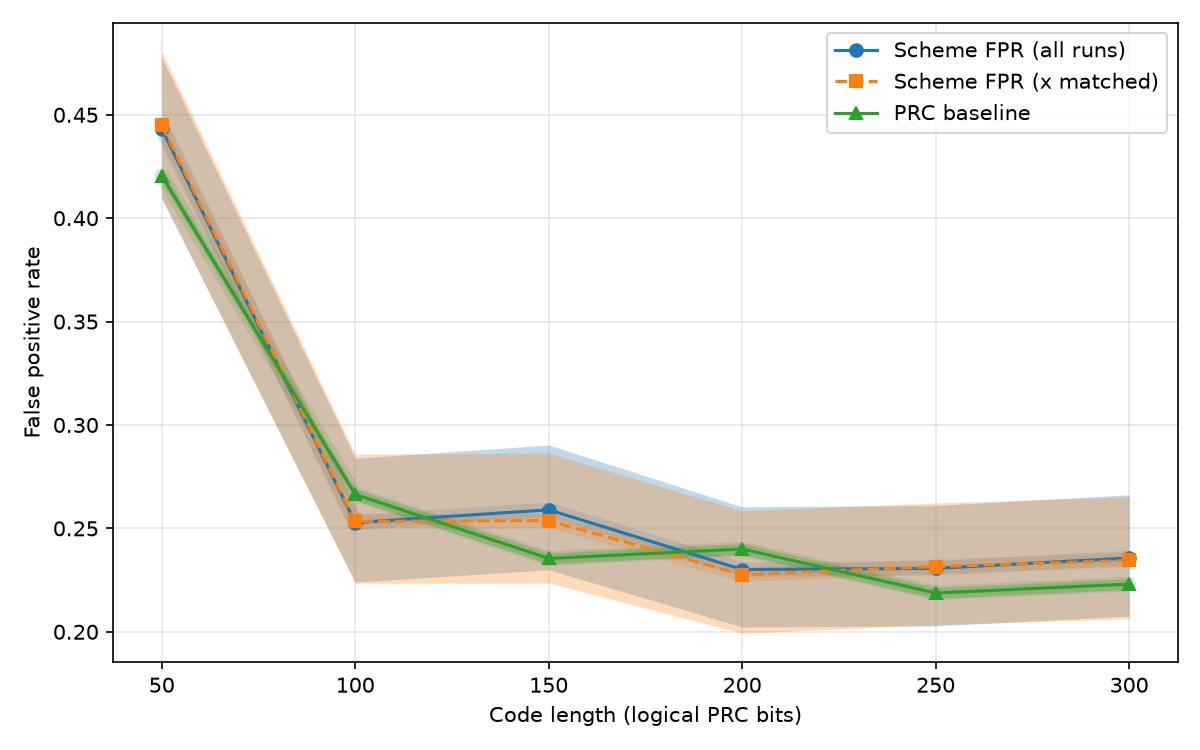}
        \caption{Empirical false-positive rate (FPR) as the output length and the corresponding PRC parameter $\kappa$ increase. Each empirical rate is estimated from 200 generated samples. We additionally report a Monte Carlo estimate of the false-positive rate of the underlying PRC detector, computed using 100,000 random samples for each value of $\kappa$. The redundancy parameter is fixed to $\rho=1$.}
        \label{fig:newfpr}
    \end{minipage}
    \hfill
    \begin{minipage}{0.47\textwidth}
        \centering
        \includegraphics[width=\linewidth]{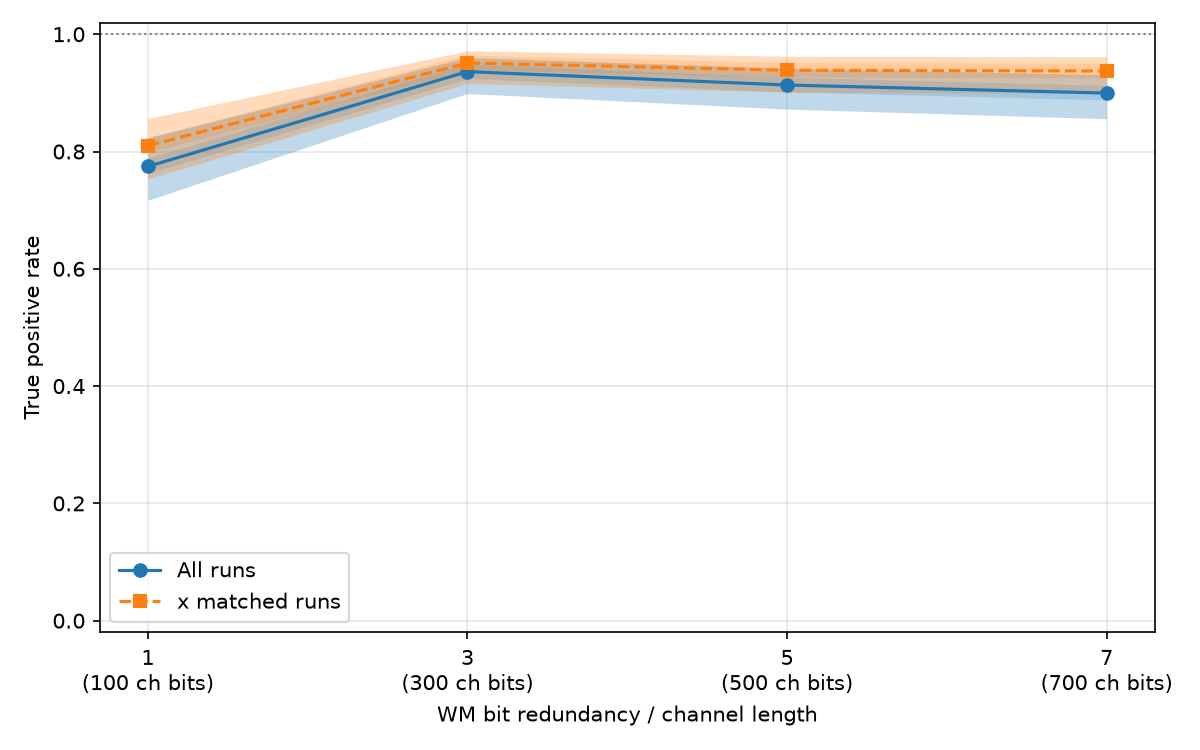}
        \caption{Empirical true-positive rate (TPR) as the output length and redundancy increase. We report both end-to-end detection, which includes errors from the multi-label attribute classifier, and detection using the intended ground-truth attribute, which bypasses attribute classification and thereby isolates watermark recovery and PRC detection under the intended policy. We fix $\kappa=100$ and evaluate 250 samples for each output length. The channel length does not include the 100 token unwatermarked prefix which is skipped by redundancy}
        \label{fig:newtpr}
    \end{minipage}
\end{figure}

\begin{table}[t]
    \centering
    \renewcommand{\arraystretch}{1.2}
    \setlength{\tabcolsep}{4pt}
    \caption{Single-attribute-prompt detection matrix. Each column $j$ corresponds to outputs generated from prompts designed to elicit target attribute $j$, and each row $i$ corresponds to a detection key constrained to attribute $i$. Each entry reports the fraction of outputs in column $j$ detected by the corresponding constrained key. Entries above $50\%$ are highlighted.}
    \label{tab:new-single-label-matrix}
    
    \begin{tabular}{cccccc}
        \hline
        & Medicine & Economics & Art & Software & Sports \\ \hline
        Medicine
        & \cellcolor{gray!25}86.0\%
        & 28.57\%
        & 26.0\%
        & 22.0\%
        & 33.33\% \\

        Economics
        & 28.0\%
        & \cellcolor{gray!25}87.76\%
        & 14.0\%
        & 26.0\%
        & 19.61\% \\

        Art
        & 24.0\%
        & 28.57\%
        & \cellcolor{gray!25}88.0\%
        & 24.0\%
        & 21.57\% \\

        Software
        & 18.0\%
        & 24.49\%
        & 16.0\%
        & \cellcolor{gray!25}98.0\%
        & 17.65\% \\

        Sports
        & 26.0\%
        & 24.49\%
        & 22.0\%
        & 22.0\%
        & \cellcolor{gray!25}88.24\% \\
        \hline
    \end{tabular}
\end{table}

\begin{table}[t]
    \centering
    \renewcommand{\arraystretch}{1.2}
    \setlength{\tabcolsep}{4pt}
    \caption{Two-attribute-prompt detection matrix. Each column corresponds to outputs generated from a prompt designed to jointly elicit two target attributes, and each row corresponds to the attribute to which the detection key is constrained. Each entry reports the fraction of outputs generated from the column prompt that are detected by the corresponding constrained key. Entries above $50\%$ are highlighted. Variation across prompts may reflect differences in both the token-level entropy of the generated responses and the accuracy of multi-label attribute classification.}
    \label{tab:new-multi-label-matrix}

    \begin{tabular}{cccc}
        \hline
        & \texttt{sports\_econ}
        & \texttt{art\_software}
        & \texttt{medicine\_software} \\ \hline

        Medicine
        & 32.0\%
        & 24.0\%
        & \cellcolor{gray!25}76.0\% \\

        Economics
        & \cellcolor{gray!25}86.0\%
        & 28.0\%
        & 26.0\% \\

        Art
        & 32.0\%
        & \cellcolor{gray!25}98.0\%
        & 26.0\% \\

        Software
        & 20.0\%
        & \cellcolor{gray!25}98.0\%
        & \cellcolor{gray!25}76.0\% \\

        Sports
        & \cellcolor{gray!25}88.0\%
        & 32.0\%
        & 28.0\% \\
        \hline
    \end{tabular}
\end{table}

\paragraph{Attribute-constrained detection.}
Table~\ref{tab:new-single-label-matrix} shows a clear separation between target and non-target attributes for the single-attribute prompts. For outputs generated from prompts targeting attribute $j$, the key constrained to $j$ achieves a detection rate between $86\%$ and $98\%$. In contrast, detection by keys constrained to non-target attributes ranges from $14\%$ to $33.33\%$. Thus, constrained keys detect outputs generated from prompts targeting their authorized attributes substantially more often than outputs generated from prompts targeting other attributes, although cross-attribute detection remains non-negligible in the current prototype, due to the baseline FPR of the PRC implementation used, as seen in Figure\ref{fig:newfpr}.

The two-attribute-prompt results in Table~\ref{tab:new-multi-label-matrix} exhibit the intended policy behavior for all three tested attribute pairs. For each prompt, both target attributes are detected at rates above $50\%$, with target-attribute detection rates ranging from $76\%$ to $98\%$. Detection under non-target attributes ranges from $20\%$ to $32\%$. These results suggest that the construction can support prompts targeting multiple attributes while retaining meaningful separation from unrelated policies.

\paragraph{Dependence on token-level entropy.}
We observe substantial variation in detection performance across prompts and subject areas. This variation is expected because the effectiveness of \textsc{RandRecovery} depends strongly on the concentration of the next-token probability distribution encountered during generation. When the distribution is less concentrated and has higher effective entropy, the token vocabulary can typically be partitioned into two sets whose probability masses are closer to $1/2$. Although the two recovered bit values need not be exactly equally likely, a more balanced partition increases the probability with which the mask can be applied, and therefore improves bit-embedding and recovery reliability. In more concentrated decoding states, a small number of tokens may carry most of the probability mass, making balanced partitions impossible to construct, drastically reducing the probability with which the desired bit can be embedded.

End-to-end detection performance additionally depends on the reliability of the multi-label attribute classifier. Classification accuracy can vary across topics and prompts: an intended target attribute may fail to exceed the threshold, or an unrelated attribute may be assigned spuriously, even when the embedded PRC codeword is otherwise recoverable. The reported detection rates therefore reflect the combined effects of the next-token probability distribution, randomness-recovery accuracy, and multi-label attribute-classification accuracy. Consequently, aggregate results may conceal substantial prompt-level heterogeneity.

\paragraph{Runtime.}
We measure the wall-clock runtime of each stage of the prototype. The
\texttt{setup} stage generates the master secret key and initializes the PRC parameters. The \texttt{no-wm gen} and \texttt{wm gen} stages generate an unwatermarked response and a watermarked response, respectively. The \texttt{issue} stage produces one unconstrained detection key and one constrained detection key for each attribute in the fixed vocabulary $\mathcal{V}$, where $|\mathcal{V}|=5$. The \texttt{detect} stage reports the runtime of a single constrained-key detection, averaged over the five constrained detection keys corresponding to the attributes in $\mathcal{V}$. Thus, the reported value is an average per-key detection time rather than the total time required to run all five constrained keys.

\begin{table}[t]
    \caption{Mean wall-clock runtime of each scheme stage, in seconds, averaged over 150 runs with $\kappa=100$.}
    \label{tab:runtime}
    \centering
    \renewcommand{\arraystretch}{1.2}
    \setlength{\tabcolsep}{4pt}

    \begin{tabular}{cccccc}
        \hline
        Generated tokens
        & Setup
        & No-WM generation
        & WM generation
        & Issue
        & Detect \\
        \hline

        400
        & 0.000143
        & 5.5573
        & 10.506
        & 0.000887
        & 4.5675 \\

        600
        & 0.000137
        & 5.8368
        & 16.437
        & 0.000953
        & 7.3855 \\

        \hline
    \end{tabular}
\end{table}

\section*{Acknowledgements and AI assistance}
The authors thank Shang-Hua Teng, Vatsal Sharan, Elaine Shi, and anonymous reviewers for useful feedback. AI was used to aid in the implementation and experiment creation significantly, primarily though AI-assisted development using Cursor and a variety of models. ChatGPT 5.5 was also utilized for LATEX table generation, as well as formatting and grammatical help.

\bibliographystyle{unsrtnat}
\bibliography{reference}

\appendix

\section{Definitions}
\subsection{Pseudorandom Error-Correcting Codes}

\begin{definition}[Pseudorandom Code \cite{christ2024pseudorandom,alrabiah2025ideal}] \label{PRCdef}
    A secret-key PRC over alphabet $\Sigma$ with threshold $\delta$ is a tuple of probabilistic polynomial time algorithms $(\Gen, \Encode, \Decode)$ with the following syntax and properties.
    \begin{itemize}
        \item \underline{Syntax.} Let $\ell_\sk, n, k$ be length functions.
        \begin{itemize}
            \item $\Gen(1^\secpara; r) \to \sk$: On input security parameter $\secpara$ and randomness $r$, it outputs a secret key $\sk \in \zo^{\ell_\sk}$.
            \item $\Encode(\sk) \to c \in \Sigma^{n}$: On input a secret key $\sk$, it outputs a codeword $c$.
            \item $\Decode(\sk, c \in \Sigma^n) \to 1 / 0$: 
            On input a secret key $\sk$ and a word $c$, it outputs $1$ if the decoding succeeds; otherwise, it outputs $0$.
        \end{itemize}
        \item \underline{$\delta$-Adaptive Robustness.} For every efficient adversary $\adv$, \[\Pr[\adv \text{ wins } \Game^{\mathsf{sk-robust}}_{\PRC,\delta,\adv}(\kappa)]\leq\negl(\secpara),\]
        where $\Game^{\mathsf{sk-robust}}_{\PRC,\delta,\adv}(\kappa)$ is the following security game: \begin{enumerate}
    \item The challenger samples $\sk\gets\PRC.\Gen(1^\secpara)$ and initialize transcript $\tau=\emptyset$.
    \item The adversary is allowed to make encoding queries. For each encoding query, the challenger responds $c=\PRC.\Encode(\sk)$ and sets $\tau \gets \tau\cup\{c\}$.
    \item The adversary sends the challenger $c^*$.
    \item The challenger computes $d \gets \PRC.\Decode(\sk,c^*)$. If there exists $c\in \tau$ such that $\dist(c^*, c) < \delta$ and $d = 0$, then the adversary $\adv$ wins; otherwise $\adv$ loses.
\end{enumerate}
        \item \underline{Soundness.} For any fixed word $c \in \Sigma^*$, 
        \[
            \pr[\sk \gets \Gen(1^\secpara)]{\Decode(\sk, c) = 0} \geq 1 - \negl(\secpara).
        \]
        \item \underline{Pseudorandomness.} For every probabilistic polynomial-time adversary $\adv$, 
        \[
            \abs{\pr[\sk \gets \Gen(1^\secpara)]{\adv^{\Encode(\sk, \cdot)}(1^\secpara)} - \pr{\adv^{\mathcal{U}}(1^\secpara)}} = \negl(\secpara),
        \]
        where the oracle $\mathcal{U}$ return a uniformly random word in $\Sigma^n$ for each query.
    \end{itemize}
\end{definition}

\subsection{Constrained Pseudorandom Functions (CPRFs)} 

\begin{definition}[Constrained Pseudorandom Functions] \label{CPRFdef}
Let $\secpara$ be security parameter and $\pp$ be a public parameter as inputs implicit in all algorithms. A Constrained Pseudorandom Function (CPRF) with key space $\mathcal{K}$, domain $\mathcal{X}$, and range $\mathcal{Y}$ that supports constraints represented by the class of circuits $\mathcal{C}$, where each $C \in \cC$ maps $\mathcal{X}$ to $\{0,1\}$, is a tuple of polynomial time algorithms $(\Gen, \Eval, \constrain, \CEval)$ with the following syntax and properties.

\begin{itemize}
    \item \underline{Syntax.}
    \begin{itemize}
        \item $\Gen(1^\secpara) \rightarrow \msk$: The randomized key generation algorithm takes as input a security parameter $\secpara$ and outputs a master secret key $\msk \in \mathcal{K}$.
        \item $\Eval(\msk,x)\rightarrow y$: The deterministic evaluation algorithm takes as input the master secret key $\msk$ and input $x \in \mathcal{X}$, and outputs $y \in \mathcal{Y}$.
        \item $\constrain(\msk, C)\rightarrow \csk$: The randomized constrain algorithm takes as input the master secret key $\msk$ and a constraint circuit $C \in \mathcal{C}$, and outputs a constrained key $\csk$.
        \item $\CEval(\csk,x)\rightarrow y$: The deterministic constrained evaluation algorithm takes as input the constrained key $\csk$ and an input $x \in \mathcal{X}$ and outputs $y \in \mathcal{Y}$.
    \end{itemize}
    \item \underline{Correctness.} For all security parameters $\secpara$, all constraints $C \in \mathcal{C}$, and all inputs $x \in \mathcal{X}$ such that $C(x)=1$ (authorized), it holds that:
    \[\Pr\left[\Eval(\msk,x)=\CEval(\csk,x)\;\middle|\; \begin{matrix}
        \msk\gets\Gen(1^\secpara)\\ \csk \gets \constrain(\msk,C)
    \end{matrix}\right]\geq 1- \negl(\secpara)\]
    \item \underline{Security.} Any efficient adversary $\adv$ has negligible distinguishing advantage
    \[\abs{\Pr[1 \gets \exp_{\adv,0}^{\mathsf{CPRF}}(\secpara)] - \Pr[1 \gets \exp_{\adv,1}^{\mathsf{CPRF}}(\secpara)]} = \negl(\secpara),\]
    where $\exp_{\adv,b}^{\mathsf{CPRF}}(\secpara)$ is the following security experiment ($b$ is the challenge bit):
\begin{enumerate}
    \item \textbf{Setup}: The challenger runs $\msk \gets \Gen(1^\secpara)$, initializes the set $Q:=\emptyset$, and runs $\adv(1^\secpara)$.
    \item \textbf{Pre-challenge queries}: $\adv$ adaptively sends arbitrary inputs $x \in \mathcal{X}$ to the challenger. For each $x$, the challenger computes $y:=\Eval(\msk,x)$, sends $y$ to $\adv$, and proceeds to update $Q:=Q \cup \{x\}$.
    \item \textbf{Constrain queries}: $\adv$ sends one constraint $C \in \mathcal{C}$ to the challenger. The challenger computes $\csk \gets \constrain(\msk,C)$, and sends $\csk$ to $\adv$.
    \item \textbf{Challenge queries}: For the single challenge query, $\adv$ sends input $x^* \in \mathcal{X}$ as its challenge query, subject to the restriction that $x^* \notin Q$ and $C(x^*)= 0$. If $b=0$, the challenger computes $y^*=\Eval(\msk, x^*)$. Else, if $b=1$, the challenger samples $y^* \gets \cY$.  The challenger sends $y^*$ to $\adv$.
    \item $\adv$ outputs a bit as its guess.
\end{enumerate}
\end{itemize}

\end{definition}

\section{Theoretical Proofs}
\label{sec:proofs}

\cref{thm:main-theorem} asserts that  \cref{construction} satisfies four key properties, which we prove one by one in the next four lemmas (\cref{lemma:consistency,lemma:adaptive-robustness,lemma:undetectability,lemma:soundness}).\footnote{For simplicity, the lemmas are stated for classifiers that are consistent and adversarially robust with negligible error. 
The latter requirement on the negligible error rate can be relaxed based on the concrete target error rate of the watermarking security.
}

\begin{lemma}
\label{lemma:consistency}
\cref{construction} satisfies consistency (\cref{def:consistency}).
\end{lemma}
\begin{proof}
    Let $f\in \cF$ and $\dk_f \gets \Issue(\msk, f) = \CPRF.\Constrain(\msk, f)$. For every $x \in \cX$, by the correctness of CPRF, we have
    \[\CPRF.\Eval(\msk, x) = \CPRF.\CEval(\dk_f, x)\] with probability $1-\negl(\secpara)$. Notice that in the construction, $\MasterDetect(\msk, \out)$ and $\Detect(\dk_f,\out)$ derive $r$ from $\CPRF.\Eval(\msk, x)$ and $\CPRF.\CEval(\dk_f, x)$ respectively using $x = \Classifier(\out)$, and performs the same deterministic computation afterwards. Therefore, $\MasterDetect(\msk, \out) = \Detect(\dk_f, \out)$ holds with probability $1-\negl(\secpara)$.
\end{proof}

\begin{lemma}
\label{lemma:adaptive-robustness}
\cref{construction} satisfies adaptive robustness with detection against a constant $\delta$ fraction of adversarial corruptions (\cref{def:adaptive robustness}).
\end{lemma}
\begin{proof}
    Let $f\in\cF$ be any policy and $\adv$ be any probabilistic polynomial-time adversary. Let $T = \poly(\secpara)$ be the running time of $\adv$. Consider the following probabilities.
    \begin{itemize}
        \item Let $p_0$ be the probability that the security game $\frakG^{\mathsf{Robust}}_{\secpara,f,\cW}(\adv)$ outputs $1$. 
        \item Let $p_1$ be the probability that the following modified security game outputs $1$. The modified security game makes the following changes.
        \begin{itemize}
            \item Additionally sample a uniform $j \in \set{1,2,\dots,T}$ and initialize a set $\cT \gets \emptyset$.
            \item For each generation query made by $\adv$ with prompt $\prompt$, additionally compute $x \gets \Classifier(\Model(\prompt))$ and check if $f(x) = 0$, the special value $\wt{x}$ is unset, and $x \not\in \cT$. If so, update $\cT \gets \cT \cup \set{x}$. If $|\cT| = j$, set a special value $\wt{x} \gets x$. The rest of the steps for handling a generation query are the same.
            \item The modified security game outputs $1$ if both the original conditions and the condition $\Classifier(\out^*) = \wt{x}$ hold.
        \end{itemize}
        
        The original security game only accepts when $\dist(\out^*, \cS) < \delta$, \ie when there exists $\out \in \cS$ such that $\dist(\out^*, \out) < \delta$. By the adversarial robustness of the classifier, we have $\Classifier(\out^*) = \Classifier(\out)$ except with negligible probability. Therefore, up to a negligible probability, the original security game only outputs $1$ when $\Classifier(\out^*) \in \cT$. Since $\wt{x}$ is chosen from the set $\cT$ according to an independently sampled $j \in \set{1,2,\dots, T}$, the conditional probability of $\Classifier(\out^*) = \wt{x}$ is at least $\frac{1}{T}$, and therefore we have $p_1 \ge \frac{p_0 - \negl(\secpara)}{T}$.

        \item Let $p_2$ be the probability that the following modified security game outputs $1$. The modified security game makes the following changes to the previous game.
        \begin{itemize}
            \item Additionally sample a secret key $\sk \gets \PRC.\Gen(1^\secpara)$.
            \item For each generation query made by $\adv$ with prompt $\prompt$, additionally compute $x \gets \Classifier(\Model(\prompt))$ and check if $f(x) = 0$ and $x = \wt{x}$. If so, compute the corresponding output as $\out \gets \Model(\prompt; \PRC.\Generate(\sk))$ instead.
            \item The same conditions are used to determine the output of the security game, except that the computation of $\MasterDetect(\msk,\out^*)$ is replaced with the following procedure:
            \begin{itemize}
            \item Compute $c \gets \RandRecover(\out)$.
            \item Compute $d \gets \PRC.
            \Dec(\sk, c)$.
            \item Output $0$ if and only if $d = 0$.
        \end{itemize}
        \end{itemize}
        Notice that by construction, $\MasterDetect(\msk,\out^*)$ runs exactly the above procedure, except that $\MasterDetect(\msk,\out^*)$ derives $s_{x}$ and uses $s_{x}$ instead of $\sk$, where $x \gets \Classifier(\out^*)$. Recall that the security game also checks $x = \wt{x}$.

        Since we have the guarantee that $f(\wt{x}) = 0$, by the security of CPRF, even though the adversary holds $\dk_f \gets \Issue(\msk, f) = \CPRF.\Constrain(\msk,f)$, the value $r_{\wt{x}} = \CPRF.\Eval(\msk,\wt{x})$ is still indistinguishable from a random element. Moreover, the pseudorandomness of $\PRG$ and the pseudorandomness of $r_{\wt{x}}$ implies that the value $\PRG(r_{\wt{x}})$ is indistinguishable from a random string. Therefore, the adversary cannot distinguish if $s_{\wt{x}} = \PRC.\Gen(1^\secpara; r_{\wt{x}})$ is changed to $\sk \gets \PRC.\Gen(1^\secpara)$, which implies $\abs{p_1 - p_2} \le \negl(\secpara)$.
        \item Using $f$ and $\adv$, we build an efficient adversary $\cB$ against the adaptive robustness property of $\PRC$, where we define $\cB$ as follows.
    \begin{itemize}
        \item Generate $\msk \gets \Setup(1^\secpara)$ and $\dk_f \gets \Issue(\msk,f)$.
        \item Sample a uniform $j \in \set{1,2,\dots,T}$ and initialize a set $\cT \gets \emptyset$.
        \item Run $\adv$ with input $\dk_f$. Whenever $\adv$ makes a generation query on prompt $\prompt$, run $x \gets \Classifier(\Model(\prompt))$. If $f(x) = 1$, compute and return $\out \gets \Generate(\msk, \prompt)$. If $f(x) = 0$, proceed as follows:
        \begin{itemize}
            \item If the special value $\wt{x}$ is still unset and if $x \not\in\cT$, update $\cT \gets \cT \cup \set{x}$. If $|\cT| = j$, set a special value $\wt{x} \gets x$.
            \item If $\wt{x} = x$, query the oracle given in the adaptive robustness security game of $\PRC$ on input $\prompt$ to obtain a codeword $c$. Return $\out \gets \Model(\prompt; c)$.\\
            Otherwise, compute and return $\out \gets \Generate(\msk,\prompt)$.
        \end{itemize} 
        \item At some point, $\adv$ outputs $\out^*$.
        \item Abort if $\Classifier(\out^*) \neq \wt{x}$.
        \item Output $c^* \gets \RandRecover(\out^*)$.
    \end{itemize}
    Notice that $\cB$ produces $\out^*$ the same way as $\adv$ does in the previous security game because the queries are handled identically.
    In addition, $\cB$ is efficient because it runs the efficient adversary $\adv$ and additional efficient algorithms. 
    
    Let $p_3$ be the probability that $\cB$ breaks the adaptive robustness of $\PRC$. We have
    \begin{align*}
        p_3 \ge & \Pr\left[ \begin{matrix}
            \exists\; c \in \cS_{\mathsf{PRC}} \text{ with } \dist(c, c^*) < 2\delta\\
            \PRC.\Decode(\sk, c^*) = 0
        \end{matrix} \right]\\
        \ge & \Pr\left[ \begin{matrix}
            \exists\; \out \in \cS \text{ with } \dist(\out, \out^*) < \delta\\
            \PRC.\Decode(\sk, c^*) = 0
        \end{matrix} \right] - \negl(\secpara)\\
        \ge & \Pr\left[ \begin{matrix}
            \wt{x} = x^* \\
            \exists\; \out \in \cS \text{ with } \dist(\out, \out^*) < \delta\\
            \PRC.\Decode(\sk, c^*) = 0
        \end{matrix} \right] - \negl(\secpara)\\
        = & \;p_2 - \negl(\secpara),
    \end{align*}
    where the first inequality follows from the robust randomness recovery property, \ie $\dist(\out, \out^*) < \delta$ implies $\dist(c, c^*) < 2\delta$ except with negligible probability, where $c$ is the first $n$ bits of the randomness used by $\Model$ to generate $\out$. The second inequality is due to the fact that the probability does not increase under the intersection of events. The equality is by the definition of $p_2$. By the adaptive robustness of $\PRC$ applied to $\cB$, we have $p_3 \le \negl(\secpara)$, and therefore $p_2 \le \negl(\secpara)$. 
    \end{itemize}
    Combining all of the above, we conclude that
    $p_0 \le T \cdot (\abs{p_1-p_2} + p_2) + \negl(\secpara) \le \negl(\secpara)$.
\end{proof}

\begin{lemma}
\label{lemma:undetectability}
    \cref{construction} satisfies undetectability (\cref{def:abundetectability}).
\end{lemma}
\begin{proof}
Let $f\in\cF$ be any policy and $\adv$ be any probabilistic polynomial-time adversary. Consider the following probabilities.
\begin{itemize}
    \item $p_0$ is the probability that $\adv^{\mathsf{Env}(1^\secpara,f)}$ outputs $1$. Recall that upon being queried a prompt $\pi$, $\mathsf{Env}(1^\secpara,f)$ will respond with an output computed as $\out \gets \Generate(\msk, \prompt)$. By construction, $\Generate(\msk,\prompt)$ performs the following:
    \begin{itemize}
        \item $x \gets \Classifier(\Model(\pi))$
        \item $r_x \gets \CPRF.\Eval(\msk,x)$
        \item $s_x \gets \PRC.\Gen(1^\secpara;\PRG(r_x))$
        \item $c \gets \PRC.\Enc(s_x,m_0)$
        \item $\out \gets \Model(\prompt;c)$
    \end{itemize}
    
    \item $p_1$ is the probability that $\adv^{\mathsf{Env}_1(1^\secpara,f)}$ outputs $1$, where we define $\mathsf{Env}_1(1^\secpara,f)$ by modifying $\mathsf{Env}(1^\secpara,f)$ in the following manner. Upon being queried a prompt $\pi$, it additionally computes $x \gets \Classifier(\Model(\pi))$ and checks if $f(x) = 0$. If so, it would proceed differently as follows:
    \begin{itemize}
        \item Maintain a set of pairs $\cP \subseteq \cX \times \cY$ initialized as an empty set $\cP \gets \emptyset$.
        \item If there is an entry $(x,y) \in \cP$, then set $r_x = y$. Otherwise, sample a uniformly random element $y \gets \cY$, set $r_x = y$, and update $\cP \gets \cP \cup \set{(x,y)}$.
        \item $s_x \gets \PRC.\Gen(1^\secpara;\PRG(r_x))$
        \item $c \gets \PRC.\Enc(s_x,m_0)$
        \item $\out \gets \Model(\prompt;c)$
    \end{itemize}
    Since $f(x) = 0$, by the security of CPRF, the value $\CPRF.\Eval(\msk,x)$ is indistinguishable from a random element against an adversary who holds $\dk_f \gets \Issue(\msk, f) = \CPRF.\Constrain(\msk,f)$. Therefore, $\abs{p_0 - p_1} \le \negl(\secpara).$
    
    \item $p_2$ is the probability that $\adv^{\mathsf{Env}_2(1^\secpara,f)}$ outputs $1$, where we define $\mathsf{Env}_2(1^\secpara,f)$ by modifying $\mathsf{Env}_1(1^\secpara,f)$ in the following manner. Upon being queried a prompt $\pi$, it additionally computes $x \gets \Classifier(\Model(\pi))$ and checks if $f(x) = 0$. If so, it would proceed differently as follows:
    \begin{itemize}
        \item Maintain a set of pairs $\cP' \subseteq \cX \times \zo^{\poly(\secpara)}$ initialized as an empty set $\cP' \gets \emptyset$.
        \item If there is an entry $(x,z) \in \cP'$, then set $s_x = z$. Otherwise, sample $z \gets \PRC.\Gen(1^\kappa)$, set $s_x = z$, and update $\cP' \gets \cP' \cup \set{(x,z)}$.
        \item $c \gets \PRC.\Enc(s_x,m_0)$
        \item $\out \gets \Model(\prompt;c)$
    \end{itemize}
    By the pseudorandomness of $\PRG$, the output of $\PRG$ under a random seed is computationally indistinguishable from a random element, so $\abs{p_1 - p_2} \le \negl(\secpara)$.
    
    \item $p_3$ is the probability that $\adv^{\mathsf{Env}_3(1^\secpara,f)}$ outputs $1$, where we define $\mathsf{Env}_3(1^\secpara,f)$ by modifying $\mathsf{Env}_2(1^\secpara,f)$ in the following manner. Upon being queried a prompt $\pi$, it additionally computes $\out' \gets \Model(\pi)$, $x \gets \Classifier(\out')$ and checks if $f(x) = 0$. If so, it would proceed differently as follows:
    \begin{itemize}
        \item Sample a fresh random string $c$, and set $\out \gets \Model(\prompt;c)$
    \end{itemize}
    By the pseudorandomness of $\PRC$, a random output of $\PRC$ under a randomly fixed secret key is computationally indistinguishable from a fresh random string. Therefore, $\abs{p_2 - p_3} \le \negl(\secpara)$.

    \item $p_4$ is the probability that $\adv^{\mathsf{Env}'(1^\secpara,f)}$ outputs $1$. Notice that the only difference between $\mathsf{Env}_3(1^\secpara,f)$ and  $\mathsf{Env}'(1^\secpara,f)$ is that when $f(x)=0$, the former samples a new $\out \gets \Model(\prompt)$, while the latter sets $\out \gets \out'$, where $\out'$ was also sampled from $\Model(\prompt)$. We see that $\out,\out'$ follow the same distribution. In computing $p_3$ and $p_4$, these outputs were additionally used to compute the attributes. Hence, the only scenario where $\mathsf{Env}_3(1^\secpara,f)$ and $\mathsf{Env}'(1^\secpara,f)$ differ is when the attributes are different. Since the classifier is $\eta$-consistent on $\Model$, we have
    \[\abs{p_3 - p_4} \le \Pr_{\out,\out' \gets \Model(\prompt)}[\Classifier(\out) \neq \Classifier(\out')] \le 1 - \eta \le \negl(\secpara)\]
\end{itemize}
By the triangle inequality, we conclude that 
\[\left|
\Pr\left[1 \gets \adv^{\mathsf{Env}(1^\secpara,f)}\right]-\Pr\left[1 \gets \adv^{\mathsf{Env}'(1^\secpara,f)}\right]\right| = \abs{p_0 - p_4} \le \negl(\secpara).\]
\end{proof}

\begin{lemma}
\label{lemma:soundness}
    \cref{construction} satisfies soundness (\cref{def:soundness}).
\end{lemma}

\begin{proof}
Let $f\in\cF$ be any policy and $\adv$ be any probabilistic polynomial-time adversary. Consider the following probabilities.
\begin{itemize}
    \item Let $p_0$ be the probability that $\adv$ wins
    the security game $\mathfrak{G}^{\mathsf{Sound}}_{\secpara,f,\cW}$, in which $\adv$ produces an output $\out^*$ and one computes $x^* \gets \Classifier(\out^*),\; c^* \gets \RandRecover(\out^*)$. $\adv$ wins if $f(x^*)=0$ and $\MasterDetect(\msk,\out^*)=1$. Recall from construction that the latter condition for $\MasterDetect$ is equivalent to computing
    \begin{itemize}
        \item $r_{x^*}\gets \CPRF.\Eval(\msk,x^*)$
        \item $s_{x^*}\gets \PRC.\Gen(1^\secpara;\PRG(r_{x^*}))$
    \end{itemize}
    and checking if $\PRC.\Dec(s_{x^*},c^*)=1$.
    
    \item Let $p_1$ be defined similarly as $p_0$, except that in the final $\MasterDetect$ computation, if $f(x^*)=0$, we replace $r_{x^*} \gets \CPRF.\Eval(\msk,x^*)$ with a uniformly random value $r^*$ in the range of $\CPRF$. By the security of the constrained PRF, from the view of an adversary $\adv$ who only receives $\dk_f \gets \Issue(\msk,f) = \CPRF.\Constrain(\msk,f)$, it cannot distinguish $r_{x^*}$ from a random value when $f(x^*)=0$, so $|p_0-p_1| \le \negl(\secpara)$.

    \item Let $p_2$ be defined similarly as $p_1$, except that in the final $\MasterDetect$ computation, if $f(x^*)=0$, we replace $s_{x^*} \gets \PRC.\Gen(1^\secpara;\PRG(r^*))$ with a freshly sampled PRC key  $s^* \gets \PRC.\Gen(1^\secpara)$. By the pseudorandomness of $\PRG$, $|p_1-p_2| \le \negl(\secpara).$ Since $s^*$ is a fresh PRC key, by the soundness of $\PRC$, for any fixed $c$, 
    \[\underset{s^*\gets \PRC.\Gen(1^\secpara)}{\Pr}\left[\PRC.\Dec(s^*,c)=1\right]\le \negl(\secpara).\] Applying this to $c^*$, we have $p_2 \le \negl(\secpara).$
\end{itemize}

Combining the above, we obtain
$
p_0
\le |p_0-p_1|+|p_1-p_2|+p_2
\le \negl(\secpara)
$.
\end{proof}

\section{Experimental Details}
\label{sec:experimental-details}

\begin{table*}[ht]
\label{single-label-prompts}
\caption{Prompts used for the single-label benchmark (\cref{tab:new-single-label-matrix}), with sample output snippets. $\kappa=100,\rho=3$.}
\footnotesize
\centering
\renewcommand{\arraystretch}{1.2}
\setlength{\tabcolsep}{4pt}

\begin{tabularx}{\textwidth}{XXXc}
\hline
\textbf{Prompt} & \textbf{Unwatermarked Generation} & \textbf{Watermarked Generation} & \textbf{Labels} \\
\hline
Explain how stem cell therapy is being used in regenerative medicine. &
Stem cell therapy is a rapidly evolving field in regenerative medicine, where stem cells are used to repair, 
replace, or regenerate damaged or diseased tissues and organs. Here's an overview of how stem cell therapy is being
used in regenerative medicine: [...truncated] &
Stem cell therapy, also known as stem cell treatment or regenerative medicine, is a type of medical treatment that 
uses stem cells to repair or replace damaged or diseased tissues and organs. Here's an overview of how [...truncated] &
medicine \\
\hline
Explain the economic effects of raising the minimum wage on employment and businesses. &
Raising the minimum wage can have both positive and negative economic effects on employment and businesses. Here 
are some of the key effects:

**Positive effects:**

1. **Increased consumer spending**: Higher minimum wages can lead to increased consumer spending, [...truncated] &
Raising the minimum wage can have both positive and negative economic effects on employment and businesses. Here 
are some of the key effects:

**Positive effects:**

1. **Increased wages for low-income workers**: Raising the minimum wage can lead to higher wages [...truncated] &
economics \\
\hline
Explain how Surrealist artists used dream imagery to challenge reality and logic. &
Surrealist artists, a group of 20th-century artists who sought to challenge reality and logic, used dream imagery as a primary 
tool to achieve this goal. By tapping into the subconscious mind and exploring the realm of the irrational, they created [...truncated] &
Surrealist artists, a group of French artists led by André Breton, used dream imagery to challenge reality and logic in their  
works, particularly in the 1920s and 1930s. This movement was characterized by the exploration of the subconscious mind, [...truncated] &
art \\
\hline
Break down the most influential software breakthroughs in history. &
Here's a breakdown of some of the most influential software breakthroughs in history:

1. **The Development of the First Computer (1936)**: Charles Babbage's Analytical Engine, designed in the 19th century, was the
first mechanical computer.  [...truncated] &
There have been numerous influential software breakthroughs throughout history. Here's a breakdown of some of the most
significant ones:

1. **Computer Algorithms (1950s-1960s)**: The development of algorithms like the Quick Sort, Merge Sort, and Binary Search [...truncated] &
software \\
\hline
Explain the role of strategy and teamwork in achieving success in sports. &
Strategy and teamwork are crucial components of success in sports, as they enable athletes to work together, adapt to different
game situations, and overcome obstacles to achieve their goals. [...truncated] &
Strategy and teamwork are crucial elements in achieving success in sports. Both are essential components that work together to 
drive performance, build relationships, and foster a winning culture. [...truncated] &
sports \\
\hline
\end{tabularx}

\end{table*}

\begin{table*}[ht]
\label{multi-label-prompts}
\caption{Prompts used for the multiple-label benchmark (\cref{tab:new-multi-label-matrix}), with sample output snippets. $\kappa=100,\rho=3$.}
\footnotesize
\centering
\renewcommand{\arraystretch}{1.2}
\setlength{\tabcolsep}{4pt}

\begin{tabularx}{\textwidth}{XXXc}
\hline
\textbf{[ID] Prompt} & \textbf{Unwatermarked Generation} & \textbf{Watermarked Generation} & \textbf{Labels} \\
\hline
\texttt{[sports\_econ]} Explain the economic nuance and impact of Drake Maye during his college football career at North Carolina. &
Drake Maye is a standout quarterback in college football, playing for the University of North Carolina Tar Heels. 
His economic nuance and impact on the team are multifaceted and significant.

**Economic Nuance:**

1. **Financial Aid:** Maye's family is considered to be of modest means, and he received a significant amount of 
financial aid to help cover the costs of attending the University of North Carolina. [...truncated] &
Drake Maye is a college football player who currently plays as a quarterback for the North Carolina Hurricanes. 
During his college career, Maye was a key player and a rising star in the ACC. Here's an analysis of the economic 
nuance and impact of his career:

**Economic Nuance:**

1. **Marketing and Branding**: Maye's college football career generated significant attention and media coverage, 
which led to increased brand awareness and exposure for the Carolina brand. [...truncated] &
economics, sports \\
\hline
\texttt{[art\_software]} Explain how software has transformed the art world. &
Software has significantly transformed the art world in various ways, revolutionizing the way artists create, 
exhibit, and interact with their work. Here are some key ways software has impacted the art world:

1. **Digital Art Creation**: Software has enabled artists to create digital artworks that can be easily shared, 
printed, and displayed. Digital art has opened up new possibilities for artists to experiment with different 
styles, techniques, and mediums. [...truncated] &
Software has revolutionized the art world in several ways, transforming the way artists create, collaborate, and 
exhibit their work. Here are some key ways software has impacted the art world:

1. **Digital Artistic Tools**: Software has enabled artists to create digital art, which can be easily shared and 
exhibited. Tools like Adobe Photoshop, Corel Painter, and Clip Studio Paint allow artists to manipulate and edit 
digital images, creating new and innovative forms of art. [...truncated] &
art, software \\
\hline
\texttt{[medicine\_software]} Explain how software has transformed the practice of medicine. &
Software has revolutionized the practice of medicine in numerous ways, transforming the way healthcare 
professionals approach diagnosis, treatment, and patient care. Here are some key ways software has impacted the 
practice of medicine: [...truncated] &
Software has revolutionized the practice of medicine in numerous ways, transforming the way healthcare 
professionals diagnose, treat, and manage patients. Here are some key examples: [...truncated] &
medicine, software \\
\hline
\end{tabularx}
\end{table*}

\paragraph{Sources of error and other limitations:}
\begin{itemize}
    \item \textbf{Classification}. The MNLI model will randomly misclassify the generated text significantly, which results in an error since the result of CPRF is affected. Using a different implementation of classification may be able to easily improve this prototype.
    \item \textbf{RandRecovery}. When token sampling has low entropy during watermarked generation, it reduces the opportunity for the provided randomness to be encoded, later reducing recovery accuracy. Additionally, since the RandRecovery implementation uses the non-special token index of each token in the text to recreate the relevant partitions and reclaim the randomness, it must identify retokenization discrepancies, and throw away more random bits to re-align itself, further impacting recovery success. This also makes the RandRecovery implementation extremely susceptible to deletion or insertion attacks.
    \item \textbf{PRC}. As shown in \cref{fig:newfpr}, the false positive rate of the scheme directly follows that of the bare PRC attempting to decode a random, unrelated sequence. Using PRC implementations with better false positive rates will directly improve the FPR of the watermarking scheme. Our theoretical results model the PRC as an abstract primitive satisfying pseudorandomness, soundness, and adaptive robustness. These guarantees should be distinguished from the security of any particular concrete implementation. Recent work by \cite{wang2025cryptanalysis} presents cryptanalytic attacks against certain LDPC-based PRC instantiations, including attacks that exploit their concrete algebraic and decoding structure. Thus, while these results do not contradict our generic construction, they highlight that a practical deployment must instantiate the scheme with a PRC whose concrete security has been carefully evaluated.
    \item \textbf{Prompt variety}. Benchmarks were repeatedly ran on a few arbitrarily picked demonstrative prompts to show feasibility of the watermarking scheme.
    \item \textbf{Sample size}. Due to computational limitations, all benchmarks were only ran with sample sizes of 100-500. Access to more computing would enable higher confidence analysis of behavior.

\end{itemize}

\end{document}